\documentclass[11pt]{article}
\usepackage{amsmath,amssymb}
\usepackage{amsthm}
\usepackage{graphicx}
\usepackage{subfig}
\usepackage{float}
\usepackage[justification=centering]{caption}

\newtheorem{theorem}{Theorem}[section]

\newtheorem{assumption}{Assumption}[section]
\newtheorem{remark}{Remark}[section]

\def\triangleq{{:=}}

\begin{document}
\title{\textbf{Optimal Tracking Performance of Control Systems with\newline
Two-Channel Constraints\footnotemark[1]} \vspace{1cm}}
\author{\textrm{\textbf{Chao-Yang~Chen$^{a,b}$, Bin~Hu$^{b}$\footnotemark[2], Zhi-Hong~Guan$^{b}$\footnotemark[3], Ming~Chi$^{b}$, Ding-Xin He$^{b}$}} \\
{\small \textsl{$^{a}$School of Information and Electrical Engineering, Hunan University of Science and Technology,}}\\
{\small \textsl{Xiangtan, 411201, P.R.China}}\\
{\small \textsl{$^b$
College of Automation, Huazhong University of Science and Technology,}}\\
{\small \textsl{Wuhan, 430074, P. R. China}}\\
}

\renewcommand{\thefootnote}{\fnsymbol{footnote}}
\footnotetext[1]{This work was partially supported by the National Natural Science Foundation of China under Grants 61503133, 61473128 and 61572208, and the Postdoctoral Science Foundation of China under Grants 2016M592449 and 2015M582224.}
\footnotetext[2]{Corresponding author. E-mail: hubinauto@mail.hust.edu.cn (B. Hu).}
\footnotetext[3]{Corresponding author. E-mail: zhguan@mail.hust.edu.cn (Z.-H. Guan).}

\date{}
\maketitle

\begin{abstract}
This paper focuses on the tracking performance limitation for a class of networked control systems(NCSs) with two-channel constraints. In communication channels, we consider bandwidth, energy constraints and additive colored Gaussian noise(ACGN) simultaneously. In plant, non-minimal zeros and unstable poles are considered; multi-repeated zeros and poles are also applicable. To obtain the optimal performance, the two-parameter controller is adopted. The theoretical results show that the optimal tracking performance is influenced by the non-minimum phase zeros, unstable poles, gain at all frequencies of the given plant, and the reference input signal for NCSs. Moreover, the performance limitation is also affected by the limited bandwidth, additive colored Gaussian noise, and the corresponding multiples for the non-minimum phase zeros and unstable poles. Additionally, the channel minimal input power constraints are given under the condition ensuring the stability of the system and acquiring system performance limitation. Finally, simulation examples are given to illustrate the theoretical results. \vspace{1em}

 \vspace{1em}
\end{abstract}

\section{Introduction}

Owning to the advantages of the NCSs over the traditional real-time control systems in information processing and decision-making, control and optimization of NCSs are rapidly developed and broadly applied \cite{ChenQiu13,Cuenca15,Guo10,Guo12,HuGuan16,Tang13}. However, system performance could be deteriorated, even leading to instability of the control plants, due to the limitations of the channel bandwidth \cite{Guan13Optimal,Rojas09}, channel capacity \cite{Braslavsky07,Freudenberg11,Rojas09}, delays \cite{Xu11,Qu16}, quantization \cite{Azuma12,Qu16}, congestions \cite{Feng14} and packet loss \cite{Bu13,Jiang16,Qiu15} and fault \cite{WangWang16} in the communication channels of NCSs. Therefore, the analysis and design of NCSs are difficult and challenging.
\par
The researches on performance of the control system attract a growing amount of interest in the control community, take \cite{Bakhtiara08,Ding10,GuanHuang16,Li15} as examples. The above literatures mainly focus on minimizing tracking error by designing optimal controllers. The objective of this paper is to reveal the quantitative relationship between the intrinsic properties of NCSs and the tracking performance limitation via feedback control. The researches on the performance of NCSs mainly focus on two aspects. On one hand, by invoking the information theory, the relationship between information entropy and system performance is studied, such as \cite{Martins08,Shingin12}. On the other hand, by using Bode and Poisson integral, another branch reveals that the performance of the close-loop systems is fundamentally constrained by the intrinsic properties of the system, such as \cite{Bakhtiara08,Ding10,Guan13Optimal,Jiang12,Rojas08}. By importing appropriate entropies and distortions, \cite{Shingin12} investigates the performance limitation for scalar systems with Gaussian disturbances, which implies that the achievable performance cannot be improved even if the maximum information constraint is relaxed to an average information constraint. \cite{Freudenberg11} discusses a lower bound on the achievable performance in a finite time and shows that this bound can be achieved by using linear strategies. \cite{Silva13} studies the control problems for discrete-time single-input linear time-invariant plants over a signal-to-noise ratio constrained channel. In \cite{Ding10}, by presenting the performance index constructed by tracking error energy, the authors investigate the optimal tracking of the NCSs with the down-link AWGN network channels. \cite{Jiang12} considers the disturbance attenuation performance to minimize the variance of the plant output in response to a Gaussian disturbance over an AWGN channel. In \cite{Guan13Optimal}, optimal tracking performance issues are studied for NCSs in the up-link channel with limited bandwidth and additive colored Gaussian noise channel.
\par
It is noted that those results above provide useful guidelines in the design of NCSs, including the design of communication channels. However, it is shown in \cite{Ding10,Guan13Optimal,Jiang15,Rojas09} that in order to obtain the optimal tracking performance, only up-link or down-link channel model is considered in the communication channel, while two-channel is often encountered in practice. In fact, there are two cases. In the first case, both the system sensor and the controller are far away from the plant. In the other case, only the controller is far away from the plant and the system sensor. The adopted model can be found in many real-world systems. For example, in the telemedicine system of robot-assisted neurosurgery, patient and robot are respectively the plant and the controller. The remote expert obtains information via the network transmission, and the instruction of the expert is then sent back to the robot via the network transmission. In addition, for leader-follower multi-agent systems \cite{HuFeng10}, provided that the position, velocity and direction information of a leader are considered as the reference signal, the controller is designed to achieve the minimal tracking error between the leader and the follower. However, owing to the structural characteristics of the follower and the communication constraint between the leader and the follower, the minimal tracking error cannot be zero. Thus the study of the relationship among the tracking performance, structural characteristics of followers and communication parameters (bandwidth and noise in this paper) will give some guidance for leader-follower multi-agent systems (such as unmanned aerial vehicle formation systems and multi-robot system) on how to achieve consensus tracking, including static consensus and dynamic consensus. Moreover, the optimal performance for two-channel communication channels is worthy of careful study in the model of NCSs. Better performance can be obtained by using a more flexible two-parameter controller\cite{Guan12Optimal}. Moreover, with the development of science and technology, two-parameter controller is also frequently used in practice in terms of aerospace\cite{Nagashio}, robotics\cite{Bingul}, power systems\cite{Campos}, etc. Meanwhile, the channel input of NCSs is often required to have an infinite power for the optimal tracking problem in \cite{Bakhtiara08,Ding10,Guan13Optimal}, which generally cannot be met in practice. Additionally, communication constraints for bandwidth and additive colored Gaussian noise should be included in the communication model, which is more realistic than the corresponding models presented in \cite{Braslavsky07,Ding10,Jiang12}. As in the real world, many practical systems resort to random reference signals. Examples include a jolting of a warship in the surf, a communication interference noise, a random fluctuation generated by turbulence for the flying missile, and a real-time random-noise tracking radar\cite{Rappaport1967,Zhang2004}. More information can refer to \cite{Ding10,Guan13Optimal,Jiang12}.
\par
The main goal of the present work is to adopt two-parameter controllers to investigate the best achievable tracking performance of networked control systems with two-channel constraints and the finite channel input power. This paper investigates the optimal tracking performance under bandwidth-limited, energy constraints and ACGN. The plant is described by the unstable and non-minimum phase system with multi-repeated poles and zeros. The reference signal is considered as random reference signals. The contributions of this paper can be summarized as follows. First, we consider both up-link and down-link channels with interference, which is more practical than most existing literatures which focus on either up-link or down-link channel models. Second, some fundamental constraints are incorporated  in the communication channels, including bandwidth, ACGN and channel input power. Third, considering that the channel input energy cannot be infinite in the real-world communication channels, this paper constructs a novel performance index, which can quantificationally characterize the properties of the tracking capability and the communication ability. Finally, the channel minimal input power constraints are given under the condition ensuring the stability of the system and acquiring system performance limitation.
\par
The rest of the paper is organized as follows. The problem formulation and preliminaries are given in section \ref{se2}. In section \ref{se3}, the main results of this paper are presented. We then proceed in Section \ref{se3} to formulate and solve the problem of optimal tracking for two-channel with bandwidth, energy constraints and ACGN. In Section \ref{se4}, some illustrative numerical examples are given. The conclusion is finally stated in Section \ref{se5}.
\par
The notation used throughout this paper is described as follows. For any complex number $z$, its complex conjugate is denoted by $\bar{z}$ . The transpose and conjugate transpose of a vector u are denoted by $u^T$ and $u^H$ respectively. The transpose and conjugate transpose of a matrix are denoted by $A^T$ and $A^H$, respectively. All the vectors and matrices in this paper are assumed to have compatible dimensions.

\section{Problem Formulation}\label{se2}

In this paper, we consider the NCSs depicted in Fig.1, where up-link and down-link channels are affected by the limited bandwidth and ACGN. Other communication constraints are not taken into account and are assumed to be ideal.
\begin{figure}[ht]
\centering
\begin{minipage}[c]{0.8\textwidth}
\centering
  \includegraphics[width=6.5cm]{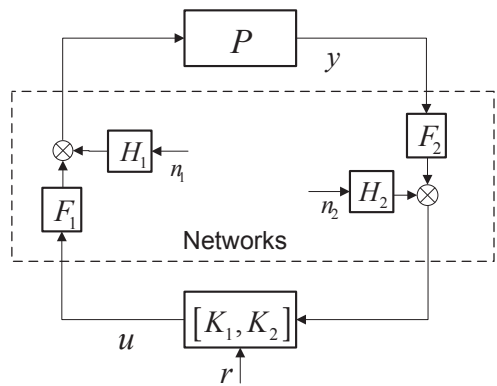}
  \caption{ Control system model of two-parameter with bandwidth-constrained additive color Gaussian noise.}\label{Fig.1}
\end{minipage}
\end{figure}
In this setup, $P$ denotes the plant model and $[K_1,K_2]$ the two-parameter controller. The down-link communication channel is characterized by three parameters: the channel transfer functions $F_1(s), H_1(s)$, and the channel noise $n_1(s)$. The up-link communication channel is considered similarly, in which $F_2(s), H_2(s)$ are the transfer functions, and $n_2(s)$ is the channel noise. The channel transfer functions $F_1(s)$ and $F_2(s)$ modeling the bandwidth limitation are assumed to be stable and nonminimum phase (NMP). $F_1(s), F_2(s)\in{\mathbb{R}H_\infty}$ have $N_{f1}$ and $N_{f2}$ NMP zeros, respectively. The channel transfer functions $H_1(s), H_2(s)\in{\mathbb{R}H_\infty}$, color the additive Gaussian noise. The communication channel additive noise process $n_i(k),(i\in\{1,2\})$ is supposed to be a zero mean stationary white Gaussian noise process, with power spectral density $\sigma_{i}^2$. The signals $r$ and $y$ are the reference input and the system output, respectively.
The reference input is a wide-sense stationary random process satisfying $r(t)=C, -\infty<t<\infty$, where $C$ is a random variable with $E[C^2]=\sigma^2$. Therefore, the power spectral density of $r(t)$ is given by $S_r(\omega)=2\pi\sigma^2\delta(\omega)$, where $\delta(\omega)$ denotes the Dirac delta function. Throughout this paper, symbols ${\bf{r},{y},{u},{n}_1}$ and ${\bf{n}_2}$ are the $\mathcal{L}$-transforms of signals $r(t), y(t), u(t), n_1(t)$ and $n_2(t)$, respectively.
\begin{assumption}\label{as2}
The system reference input $r$, channel noises $n_1$ and $n_2$ are uncorrelated.
\end{assumption}
Assumption \ref{as2}, is a common assumption and is adopted in \cite{Ding10,Jiang12,Guan13Optimal}. It can be relaxed at the expense of more complex expressions.
\par
For the given reference signal $r$, a tracking error of the system is defined as
\begin{equation}\label{eq1}
    e=r-y.
\end{equation}
The channel input is required to satisfy the power constraints
\begin{equation}\label{eqa1}
E\{\left\| {\bf u} \right\|_2^2\}<{\Gamma _u}~\text{and}~ E\{\left\| {\bf y} \right\|_2^2\} < {\Gamma _y},
\end{equation}
where $\Gamma _u$ and $\Gamma _y$ are the control input channel power constraint and the system output channel power constraint, respectively.
\begin{remark}
If the noise variance is zero or the input is unconstrained, the capacity of the channel is infinite. The common limitation on the input is a channel energy or power constraint, which can be assumed as an average power constraint\cite{Cover2012}.
Similar to the literature \cite{Braslavsky07}, we define the control input channel power and the system output channel power $E\{\left\| {\bf u} \right\|_2^2\}$ and $E\{\left\| {\bf y} \right\|_2^2\}$, respectively.
\end{remark}
A power constraint such as (\ref{eqa1}) may arise either from electronic hardware limitations or from regulatory constraints introduced to minimize interference to other communication system users.
\par

The comprehensive performance of the system is defined as
\begin{align}\label{eq2}
J \triangleq &{\varepsilon _1}E\left\{ {\left\| {\bf e} \right\|_2^2}\right\}
+ {\varepsilon _2} E\left\{\left\| {\bf u} \right\|_2^2 - {\Gamma _u}\right\} + {\varepsilon _3} E\left\{\left\| {\bf y} \right\|_2^2 - {\Gamma _y} \right\}.
\end{align}
\begin{remark}
In this performance index (\ref{eq2}), the first part reflects the tracking performance of the system, the second and third parts reflect the performance of channel communications. Therefore, by invoking these three parts tradeoff in the performance index, we can characterize the properties of the tracking capability and the communication ability. Compared with \cite{Ding10,Guan12Optimal} which only considered the tracking error energy, the performance index can better reflect communication capabilities and system tracking capabilities by weighted additional performance index (\ref{eq2}). Additionally, the proportion can be adjusted according to actual needs by weight factors.
\end{remark}
The optimal tracking performance is measured by the possible minimal tracking error achievable by all possible linear stabilizing controllers (denoted by $U$), determined as
$${J^*} = \mathop {\inf }\limits_{K \in U} J.$$
Next we introduce some important factorizations that will be frequently used in the development of the result. First, let the coprime factorization of ${F_2}P{F_1}$ be given by
\begin{equation}\label{eq3}
{F_2}P{F_1} = N{M^{-1}},
\end{equation}
where $N, M\in{\mathbb{R}H_\infty}$ and satisfy the Bezout identity
\begin{equation}\label{eq4}
XM-YN=1,
\end{equation}
for some $X, Y\in{}\mathbb{R}H_\infty$. Owing to channel transfer functions $F_1(s)$ and $F_2(s)$ are stable and NMP transfer functions, the coprime factorization of $P$ can be given by
$P=\hat{N}M^{-1},$
where $\hat{N}\in{\mathbb{R}H_\infty}$.
It is useful to factorize $\hat{N}, F_1, F_2, N$ and $M$ as
\begin{align}\label{eq6}
\begin{array}{l}
\hat{N}=L_{g}N_0,~F_1=L_{f_1}{F_{1_0}},~F_2=L_{f_2}{F_{2_0}},
N=LN_m=LF_{1_0}F_{2_0}N_0,~M=BM_m,
\end{array}
\end{align}
where $N_0(s), N_m(s)$ and $M_m(s)$ are the minimum phase transfer functions. It is easy to see that $N$ contain the NMP poles of the plant $P$ and transfer functions $F_1$  and $F_2$, but $\hat{N}$ only contain the NMP poles of the plant $P$. And $L(s), B(s)$ represent all-pass factor which can be constructed as
\begin{align*}
L(s)=&\prod_{i=1}^{N_z+N_{f_1}+N_{f_2}}\limits\left(\frac{s-z_{i}}{s+\bar{z}_{i}}\right)^{m_i}\\
=&\prod_{j=1}^{N_z}\limits\left(\frac{s-z_{j}}{s+\bar{z}_{j}}\right)^{m_j}
\prod_{k=N_z+1}^{N_z+N_{f_1}}\limits\left(\frac{s-z_{i}}{s+\bar{z}_{i}}\right)^{m_k}
\prod_{l=N_z+N_{f_1}+1}^{N_z+N_{f_1}+N_{f_2}}\limits
\left(\frac{s-z_{l}}{s+\bar{z}_{l}}\right)^{m_l},\\
B(s)=&\prod_{i=1}^{N_p}\limits\left(\frac{s-p_i}{s+\bar{p}_i}\right)^{n_i},
\end{align*}
where $z_{i},(i=1, \cdots ,N_z+N_{f_1}+N_{f_2})$ are the non-minimum phase zeros of $F_2PF_1$.
$z_{j}, z_{k},$ and $ z_{l},$ $(j = 1, \cdots ,N_z;\,k = N_z+1, \cdots ,N_z+N_{f_1};\,l = N_z+N_{f_1}+1, \cdots ,N_z+N_{f_1}+N_{f_2})$,  are the non-minimum phase zeros of $P, F_1, F_2$, respectively. $p_i, (i = 1, \cdots ,N_p)$ are the unstable poles of $P$. $m_{j}, m_{k},$ and $ m_{l},$ $(j = 1, \cdots ,N_z;\,k = N_z+1, \cdots ,N_z+N_{f_1};\,l = N_z+N_{f_1}+1, \cdots ,N_z+N_{f_1}+N_{f_2})$, are the corresponding non-minimum phase zeros multiplicity $P, F_1$, and $F_2$. $n_i (i = 1, \cdots ,N_p)$ are the corresponding unstable poles' multiplicity $P$. It is well-known that any stabilizing compensator $K$ can be described via the so-called Youla parameterization\cite{Vidyasagar}. Then, the set of all stabilizing compensators $K$ is characterized by the set\cite{Ding10,Guan13Optimal}
\begin{align}
\mathcal{K}:=\Big\{K: K=\big[K_1~K_2\big]=(X-RN)^{-1}
\big[Q~~Y-RM\big],~Q,R\in{}\mathbb{R}H_\infty\Big\}.
\end{align}
\begin{remark}
When the up-link channel and down-link channel are subject to communication constraints, the close-loop system constituted by the feedback controller and $F_2PF_1$ is internal stability.
\end{remark}

\section{Tracking performance limitations}\label{se3}

In this section, we study the optimal tracking performance over bandwidth and energy constraint channels with additive colored Gaussian noise, as shown in Figure 1. Our main result in this paper is the following theorem that provides an exact expression on the optimal tracking performance.

\begin{theorem}\label{th1}
Consider the network control system with the structure model shown in Figure 1, assumption 1. Reference signal $r$ is a random variable with zero mean and variance $\sigma_r^2$. Channel noises $n_1$ and $n_2$ are white Gaussian signals. The system $P$ is supposed to be unstable, NMP, strictly proper, transfer function. Denote $z_{i},(i=1,\cdots,N_z)$ and ${p_{i}},\left( {i = 1, \cdots ,{N_p}} \right)$ are the unstable poles and NMP zeros of the system $P$, respectively. Suppose $F_1$ and $F_2$ are NMP transfer functions. Denote $z_i,(i=N_z+1,\cdots,N_z+N_{f_{1_i}})$ and $z_i,(i=N_z+N_{f_{1_i}}+1,\cdots,N_z+N_{f_{1_i}}+N_{f_{2_i}})$, are the NMP zeros of the transfer functions $F_1$ and $F_2$.\par
Then
\begin{align*}
{J^*}=& 2{\varepsilon _1}{\sigma_r ^2}{\sum\limits_{i = 1}^{{N_z+N_{f_1}}} {{\mathop{\rm Re}\nolimits} \left\{ {{z_{i}}} \right\}}} + \sum\limits_{i = 1}^{N_{p}}\sum\limits_{d = 1}^{n_{i}} {\frac{r_{pid}}{(d - 1)!}}\sum\limits_{j = 1}^{N_{p}}\sum\limits_{k = 1}^{n_{j}} {\frac{{\rm{d}}^{d-1}}{{\rm{d}}{s}^{d-1}}\frac{(-1)^{k-1}{\bar{r}_{pjk}}}{(s + {\bar{p}_{j}})^k}}\Big|_{s=p_i}\\
&
\sum\limits_{i = 1}^{N_{z}+N_{f_1}+N_{f_2}}\sum\limits_{d = 1}^{m_{i}} {\frac{r_{zid}}{(d - 1)!}}\sum\limits_{j = 1}^{N_{z}+N_{f_1}+N_{f_2}}\sum\limits_{k = 1}^{m_{j}}
\frac{{\rm{d}}^{d-1}}{{\rm{d}}{s}^{d-1}}{\frac{(-1)^{k-1}{\bar{r}_{zjk}}}{(s + {\bar{z}_{j}})^k}}\Big|_{s=z_i}\\
&+ \sum\limits_{i = 1}^{N_{s}}\sum\limits_{d = 1}^{o_{i}} {\frac{r_{sid}}{(d - 1)!}}\sum\limits_{j = 1}^{N_{s}}\sum\limits_{k = 1}^{o_{j}}{\frac{{\rm{d}}^{d-1}}{{\rm{d}}{s}^{d-1}}\frac{(-1)^{k-1}{\bar{r}_{sjk}}}{(s + {\bar{s}_{j}})^k}}\Big|_{s=s_i}\nonumber\\
&+ \Big\|(I-\Delta_i\Delta_i^H)\left( {\begin{array}{*{20}{c}}
{{\Gamma _1}}\\
{{\Gamma _2}}
\end{array}} \right)\Big\|_2^2
+\sigma_r^2\Upsilon_1 - {\varepsilon _2}{\Gamma _u} - {\varepsilon _3}{\Gamma _y},
\end{align*}
where
\begin{align*}
r_{z\,id}&= \frac{{\sigma _1}\sqrt {{\varepsilon _1} + {\varepsilon _3}} }{(m_i-d)!}\frac{{\rm{d}}^{m_i-d}}{{\rm{d}}s^{m_i-d}}\Big((s-z_{i})^{m_i}
{N_0}(s){H_1}({s})M^{-1}(s){L^{ - 1}}(s)\Big)\Big|_{s=z_{i}},\\
r_{p\,id}&= \frac{-1}{(n_i-d)!}\frac{{\rm{d}}^{n_i-d}}{{\rm{d}}s^{n_i-d}}\Big((s-p_{i})^{n_i}
\Omega_o ({s})N^{-1}(s){B^{ - 1}}(s)\Big)\Big|_{s=p_{i}},\\
r_{s\,id}&= \frac{1}{(o_i-d)!}\frac{{\rm{d}}^{o_i-d}}{{\rm{d}}s^{o_i-d}}\Big((s-s_{i})^{o_i}
\Delta _i^{H}(s)\left( {\begin{array}{*{20}{c}}
{{\Gamma _1(s)}}\nonumber\\
{{\Gamma _2(s)}}
\end{array}} \right)\Big)\Big|_{s=s_{i}},\\
\Upsilon_1&=\left\|\begin{array}{*{20}{c}}
\sqrt{\varepsilon_1}(I-\varepsilon_1F_{1_0}N_0\Lambda_0^{-1}\Lambda_0^{-H}N_0^HF_{1_0}^H)\\
\sqrt{\varepsilon_2}M_m\Lambda_0^{-1}\Lambda_0^{-H}N_0^HF_{1_0}^H\\
\sqrt{\varepsilon_2}F_{1_0}N_0\Lambda_0^{-1}\Lambda_0^{-H}N_0^HF_{1_0}^H
\end{array}\right\|_2^2.
\end{align*}
\end{theorem}

\begin{proof}
The transfer functions from $n_1, n_2$ and $r$ to $u$ and $y$ are formed as
\begin{align*}
{\bf u} &= {K_1}{\bf r} + {K_2}{H_2}{\bf n_2} + {K_2}{F_2}{\bf y},{\rm{   }}\\
{\bf y} &= P\left( {{H_1}{{\bf n}_1} + {F_1}{\bf u}} \right).
\end{align*}
According to the architecture of Figure 1, we can obtain
\begin{align}
\nonumber{\bf y}=& P{\left( {I - {K_2}{F_2}P{F_1}} \right)^{ - 1}}
\left( {{F_1}{K_1}{\bf r} + {H_1}{{\bf n}_1} + {F_1}{K_2}{H_2}{{\bf n}_2}} \right),\\
\nonumber{\bf u}=& {\left( {I - {K_2}{F_2}P{F_1}} \right)^{ - 1}}
\left( {{K_1}{\bf r} + {K_2}{F_2}P{H_1}{{\bf n}_1} + {K_2}{H_2}{{\bf n}_2}} \right),\\
{\bf e}=& \Big(I-P{\left( {I - {K_2}{F_2}P{F_1}} \right)^{ - 1}}{F_1}{K_1}\Big){\bf r}
-P{\left( {I - {K_2}{F_2}P{F_1}} \right)^{ - 1}}\left( { {H_1}{{\bf n}_1} + {F_1}{K_2}{H_2}{{\bf n}_2}} \right).\label{eq7}
\end{align}
From (\ref{eq1}), (\ref{eq3}), (\ref{eq6}) and (\ref{eq7}), the performance index $(\ref{eq2})$ can be written as
\begin{align*}
J &={\varepsilon _1}E\Big( \left\| {SP{H_1}{{\bf n}_1}} \right\|_2^2 + \left\| {SP{F_1}{K_2}{H_2}{{\bf n}_2}} \right\|_2^2
 + \left\| {\left( {I - SP{F_1}{K_1}} \right){\bf r}} \right\|_2^2 \Big)\\
&  + {\varepsilon _2}E\Big( \left\| {S{K_2}{F_2}P{H_1}{{\bf n}_1}} \right\|_2^2
   + \left\| {S{K_2}{H_2}{{\bf n}_2}} \right\|_2^2 + \left\| {S{K_1}{\bf r}} \right\|_2^2 \Big)\\
 &+ {\varepsilon _3}E\Big( \left\| {SP{H_1}{{\bf n}_1}} \right\|_2^2
 + \left\| {SP{F_1}{K_2}{H_2}{{\bf n}_2}} \right\|_2^2
  + \left\| {SP{F_1}{K_1}{\bf r}} \right\|_2^2 \Big) - {\varepsilon _2}{\Gamma _u} - {\varepsilon _3}{\Gamma _y},
\end{align*}
where $S = {\left( {I - {K_2}{F_2}P{F_1}} \right)^{ - 1}}{\rm{ }}$.\par
According to (\ref{eq3})-(\ref{eq7}), the transfer function $S$ has
\begin{align*}
S =& {\left( {I - {K_2}{F_2}P{F_1}} \right)^{ - 1}}{\rm{ }}\\
 =& {\left( {I - {K_2}N{M^{ - 1}}} \right)^{ - 1}}\\
 =& M{\left( {M - {{\left( {X - RN} \right)}^{ - 1}}\left( {Y - RM} \right)N} \right)^{ - 1}}\\
 =& M{\left( {\left( {X - RN} \right)M - \left( {Y - RM} \right)N} \right)^{ - 1}}\left( {X - RN} \right)\\
 =& M{\left( {XM - YN} \right)^{ - 1}}\left( {X - RN} \right)\\
 =& M\left( {X - RN} \right).
\end{align*}
Then,
\begin{align}
J =&\nonumber \left( {{\varepsilon _1} + {\varepsilon _3}} \right)\left( {E\left\| {SP{H_1}{{\bf n}_1}} \right\|_2^2 + E\left\| {SP{F_1}{K_2}{H_2}{{\bf n}_2}} \right\|_2^2} \right)
 + {\varepsilon _2}\Big( E\left\| {S{K_2}{F_2}P{H_1}{{\bf n}_1}} \right\|_2^2
 + E\left\| {S{K_2}{H_2}{{\bf n}_2}} \right\|_2^2 \Big)\\
 &\nonumber - {\varepsilon _2}{\Gamma _u} - {\varepsilon _3}{\Gamma _y} + {\varepsilon _1}E\left\| {\left( {I - SP{F_1}{K_1}} \right){\bf r}} \right\|_2^2
 + {\varepsilon _2}E\left\| {S{K_1}{\bf r}} \right\|_2^2
 + {\varepsilon _3}E\left\| {SP{F_1}{K_1}{\bf r}} \right\|_2^2\\
=&\nonumber \left( {{\varepsilon _1} + {\varepsilon _3}} \right)\Big( \left\| {\hat N\left( {X - RN} \right){H_1}{\sigma _1}} \right\|_2^2+ \Big\| {F_1}\hat N( Y
 - RM ){H_2}{\sigma _2} \Big\|_2^2 \Big)
 + {\varepsilon _2}\Big( \left\| {{F_2}\hat N\left( {Y - RM} \right){H_1}{\sigma _1}} \right\|_2^2\\
 &\nonumber + \Big\| {M\left( {Y - RM} \right){H_2}{\sigma _2}} \Big\|_2^2 \Big) + {\varepsilon _2}\Big\| {MQ\sigma_r } \Big\|_2^2
 + {\varepsilon _1}\Big\| {\left( {I - {F_1}\hat NQ} \right)\sigma_r } \Big\|_2^2 + {\varepsilon _3}\Big\| {{F_1}\hat NQ\sigma_r } \Big\|_2^2
 - {\varepsilon _2}{\Gamma _u} - {\varepsilon _3}{\Gamma _y}\\
=\nonumber & \left\| {{\sigma _1}\sqrt {{\varepsilon _1} + {\varepsilon _3}} {N_0}{H_1}\left( {X - RN} \right)} \right\|_2^2
 + \left\| {\left( {\begin{array}{*{20}{c}}
{{\sigma _2}\sqrt {{\varepsilon _1} + {\varepsilon _3}} {F_{{1_0}}}{N_0}{H_2}}\\[5pt]
{{\sigma _1}\sqrt {{\varepsilon _2}} {F_{{2_0}}}{N_0}{H_1}}\\[5pt]
{{\sigma _2}\sqrt {{\varepsilon _2}} {M_m}{H_2}}
\end{array}} \right)\left( {Y - RM} \right)} \right\|_2^2\\
&\nonumber + \left\|\left(\begin{array}{*{20}{c}}I\\0\\0\end{array}\right)+\left( {\begin{array}{*{20}{c}}
{ - {F_1}{\hat{N}}}\\
{\sqrt {\frac{{{\varepsilon _2}}}{{{\varepsilon _1}}}} {M_m}}\\[5pt]
{\sqrt {\frac{{{\varepsilon _3}}}{{{\varepsilon _1}}}} {F_{1_0}}{N_0}}
\end{array}} \right)Q \right\|_2^2{\varepsilon _1}{\sigma_r ^2}
 - {\varepsilon _2}{\Gamma _u} - {\varepsilon _3}{\Gamma _y}\\
\buildrel \Delta \over =& {J_1} + {J_2} + {J_3},\label{eq8}
\end{align}
where
\begin{align*}
J_1=&\left\| {{\sigma _1}\sqrt {{\varepsilon _1} + {\varepsilon _3}} {N_0}{H_1}\left( {X - RN} \right)} \right\|_2^2,\\
J_2=&\left\| {\left( {\begin{array}{*{20}{c}}
{{\sigma _2}\sqrt {{\varepsilon _1} + {\varepsilon _3}} {F_{1_0}}{N_0}{H_2}}\\[5pt]
{{\sigma _1}\sqrt {{\varepsilon _2}} {F_{2_0}}{N_0}{H_1}}\\[5pt]
{{\sigma _2}\sqrt {{\varepsilon _2}} {M_m}{H_2}}
\end{array}} \right)\left( {Y - RM} \right)} \right\|_2^2,\\
J_3=&\left\|\left(\begin{array}{*{20}{c}}I\\0\\0\end{array}\right)+\left( {\begin{array}{*{20}{c}}
{ - {F_1}{\hat{N}}}\\
{\sqrt {\frac{{{\varepsilon _2}}}{{{\varepsilon _1}}}} {M_m}}\\[5pt]
{\sqrt {\frac{{{\varepsilon _2}}}{{{\varepsilon _1}}}} {F_{1_0}}{N_0}}
\end{array}} \right)Q \right\|_2^2{\varepsilon _1}{\sigma ^2}
 - {\varepsilon _2}{\Gamma _u} - {\varepsilon _3}{\Gamma _y}.
\end{align*}
Let $J_{12}=J_1+J_2$, we have
\begin{align}
J^*=&\inf_{K\in{U}}J\nonumber\\
=&\inf_{R\in{RH_\infty}}(J_1+J_2)+\inf_{Q\in{RH_\infty}}J_3\label{eqq8}\\
\buildrel \Delta \over =&J_{12}^*+J_3^*.\nonumber
\end{align}
where $J_{12}^*=\inf_{R\in{RH_\infty}}(J_1+J_2), J_3^*=\inf_{Q\in{RH_\infty}}J_3$.
\par
Firstly, in order to obtain $J_{12}^*$, we first handle $J_1$ and $J_2$, respectively.
\par
For $J_1$, consider the fact that the term
$${\sigma _1}\sqrt {{\varepsilon _1} + {\varepsilon _3}} {N_0}(s){H_1}(s)X(s){L^{ - 1}}(s),$$
can be decomposed as
\begin{equation}\label{eq9}
{\sigma _1}\sqrt {{\varepsilon _1} + {\varepsilon _3}} {N_0}(s){H_1}(s)X(s){L^{ - 1}}(s) = \Gamma _1^ \bot (s) + {\Gamma _1}(s),
\end{equation}
where ${\Gamma _1}(s)\in {H_2},{\rm{ }}\Gamma _1^ \bot (s)\in H_2^ \bot $, and
\begin{align}
\Gamma _1^ \bot (s) &= \sum\limits_{i = 1}^{N_{z}+N_{f_1}+N_{f_2}}\sum\limits_{d = 1}^{m_{i}} {\frac{{r_{z\,id}}}{(s - {z_{i}})^d}},\nonumber\\
r_{z\,id}&= \frac{{\sigma _1}\sqrt {{\varepsilon _1} + {\varepsilon _3}} }{(m_i-d)!}\frac{{\rm{d}}^{m_i-d}}{{\rm{d}}s^{m_i-d}}\Big((s-z_{i})^{m_i}
{N_0}(s){H_1}({s})X(s){L^{ - 1}}(s)\Big)\Big|_{s=z_{i}}.\label{eq11}
\end{align}
Based on the Bezout identity (\ref{eq4}) and (\ref{eq11}), we can obtain
\begin{align*}
r_{z\,id}&= \frac{{\sigma _1}\sqrt {{\varepsilon _1} + {\varepsilon _3}} }{(m_i-d)!}\frac{{\rm{d}}^{m_i-d}}{{\rm{d}}s^{m_i-d}}\Big((s-z_{i})^{m_i}{N_0}(s)
{H_1}({s})M^{-1}(s){L^{ - 1}}(s)\Big)\Big|_{s=z_{i}}.
\end{align*}
Therefore, we have
\begin{align}
J_1=&\left\| {{\sigma _1}\sqrt {{\varepsilon _1} + {\varepsilon _3}} {N_0}{H_1}\left( {X - RN} \right)} \right\|_2^2\nonumber\\
=&\| {\sigma _1}\sqrt {{\varepsilon _1} + {\varepsilon _3}} {N_0}{H_1}\left( {XL^{-1} - RN_m} \right) \|_2^2\nonumber\\
=&\Big\| {\sigma _1}\sqrt {{\varepsilon _1} + {\varepsilon _3}} {N_0}{H_1}XL^{-1}
 - {\sigma _1}\sqrt {{\varepsilon _1} + {\varepsilon _3}} {N_0}{H_1}RN_m \Big\|_2^2\nonumber\\
=&\left\| \Gamma_1^{\bot} + \Gamma_1
- {\sigma _1}\sqrt {{\varepsilon _1} + {\varepsilon _3}} {N_0}{H_1}RN_m \right\|_2^2\nonumber\\
=& \left\| {\Gamma _1^ \bot } \right\|_2^2 + \Big\| {{\Gamma _1} - {\sigma _1}\sqrt {{\varepsilon _1} + {\varepsilon _3}} {N_0}{H_1}R{N_m}} \Big\|_2^2\nonumber\\
=&\Big\|\sum\limits_{i = 1}^{N_{z}+N_{f_1}+N_{f_2}}\sum\limits_{d = 1}^{m_{i}} {\frac{{r_{zid}}}{(s - {z_{i}})^d}}\Big\|_2^2
 + \Big\| {{\Gamma _1} - {\sigma _1}\sqrt {{\varepsilon _1} + {\varepsilon _3}} {N_0}{H_1}R{N_m}} \Big\|_2^2\nonumber
\end{align}
\begin{align}
=&\sum\limits_{i = 1}^{N_{z}+N_{f_1}+N_{f_2}}\sum\limits_{d = 1}^{m_{i}} {\frac{r_{zid}}{(d - 1)!}}\sum\limits_{j = 1}^{N_{z}+N_{f_1}+N_{f_2}}\sum\limits_{k = 1}^{m_{j}}
\frac{{\rm{d}}^{d-1}}{{\rm{d}}{s}^{d-1}}{\frac{(-1)^{k-1}{\bar{r}_{zjk}}}{(s + {\bar{z}_{j}})^k}}\Big|_{s=z_i} + \Big\| {\Gamma _1} - {\sigma _1}\sqrt {{\varepsilon _1} + {\varepsilon _3}}
{N_0}{H_1}R{N_m} \Big\|_2^2\label{a10}
\end{align}
For $J_2$, we perform an inner-outer factorization given in
\cite{Vidyasagar}, such that
\begin{align}\label{tj13}
\left( {\begin{array}{*{20}{c}}
{{\sigma _2}\sqrt {{\varepsilon _1} + {\varepsilon _3}} {F_{1o}}(s){N_0}(s){H_2}(s)}\\[5pt]
{{\sigma _1}\sqrt {{\varepsilon _2}} {F_{2o}}(s){N_0}(s){H_1}(s)}\\[5pt]
{{\sigma _2}\sqrt {{\varepsilon _2}} {M_m}(s){H_2}(s)}
\end{array}} \right) = {\Omega_i}{\Omega _0},
\end{align}
where ${\Omega _i}$ and ${\Omega _0}$ are the inner and the outer.
\par
Then, we have
\begin{align*}
J_2&=\left\| {\Omega _i\Omega _0 \left( {Y - RM} \right)} \right\|_2^2\\
   &=\left\| {\Omega _0 \left( {Y - RM} \right)} \right\|_2^2\\
   &=\left\| {\Omega _0 \left( {YB^{-1} - RM_m} \right)} \right\|_2^2.
\end{align*}

Similarly to $(\ref{eq9})$, $\Omega_o (s)Y(s){B^{ - 1}}(s)$ also can be decomposed as
\begin{equation*}
\Omega_o (s)Y(s){B^{ - 1}}(s) = \Gamma _2^ \bot (s) + {\Gamma _2}(s),
\end{equation*}
where ${\Gamma _2}(s) \in {H_2},\Gamma _2^ \bot (s) \in H_2^ \bot $, and
\begin{align}
{\rm{ }}\Gamma _2^ \bot (s)&= \sum\limits_{i = 1}^{N_p}\sum\limits_{d = 1}^{n_{i}} {\frac{{r_{p\,id}}}{(s - {p_{i}})^d}},\nonumber\\
r_{p\,id}&= \frac{1}{(n_i-d)!}\frac{{\rm{d}}^{n_i-d}}{{\rm{d}}s^{n_i-d}}\Big((s-p_{i})^{n_i}
\Omega_o ({s})Y({s}){B^{ - 1}}(s)\Big)\Big|_{s=p_{i}}.\label{eq111}
\end{align}
\par
Based on the Bezout identity (\ref{eq4}) and (\ref{eq111}), we can obtain
\begin{align*}
r_{p\,id}&= \frac{1}{(n_i-d)!}\frac{{\rm{d}}^{n_i-d}}{{\rm{d}}s^{n_i-d}}\Big((s-p_{i})^{n_i}
\Omega_o ({s})N^{-1}(s){B^{ - 1}}(s)\Big)\Big|_{s=p_{i}}.
\end{align*}
Therefore, we have
\begin{align}
J_2=&\left\| {\Omega _0 \left( {YB^{-1} - RM_m} \right)} \right\|_2^2\nonumber\\
=&\left\| \Gamma_2^{\bot} + \Gamma_2 - \Omega _0RM_m \right\|_2^2\nonumber\\
=&\left\| {\Gamma _2^ \bot } \right\|_2^2 + \Big\| {{\Gamma _2} - \Omega_o R{M_m}} \Big\|_2^2\nonumber\\
=&\left\|\sum\limits_{i = 1}^{N_{p}}\sum\limits_{d = 1}^{n_{i}} {\frac{{r_{pid}}}{(s - {p_{i}})^d}}\right\|_2^2 +  \Big\| {{\Gamma _2} - \Omega_o R{M_m}} \Big\|_2^2\nonumber\\
=& \sum\limits_{i = 1}^{N_{p}}\sum\limits_{d = 1}^{n_{i}} {\frac{r_{pid}}{(d - 1)!}}\sum\limits_{j = 1}^{N_{p}}\sum\limits_{k = 1}^{n_{j}} {\frac{{\rm{d}}^{d-1}}{{\rm{d}}{s}^{d-1}}\frac{(-1)^{k-1}{\bar{r}_{pjk}}}{(s + {\bar{p}_{j}})^k}}\Big|_{s=p_i}
 + \Big\| {{\Gamma _2} - \Omega_o R{M_m}} \Big\|_2^2.\label{a11}
\end{align}
\par
Based on the above analysis, we now consider $J_{12}$, noting (\ref{a10},\ref{a11}), we have
\begin{align}
{J_{12}} =&\left\| {\Gamma _1^ \bot } \right\|_2^2 + \Big\| {{\Gamma _1} - {\sigma _1}\sqrt {{\varepsilon _1} + {\varepsilon _3}} {N_0}{H_1}R{N_m}} \Big\|_2^2
 + \left\| {\Gamma _2^ \bot } \right\|_2^2 + \Big\| {{\Gamma _2} - \Omega_o R{M_m}} \Big\|_2^2 \nonumber\\
 =& \left\| {\Gamma _1^ \bot } \right\|_2^2 + \left\| {\Gamma _2^ \bot } \right\|_2^2
 + \left\| {\begin{array}{*{20}{c}}
{{\Gamma _1} - {\sigma _1}\sqrt {{\varepsilon _1} + {\varepsilon _3}} {N_0}{H_1}R{N_m}}\\
{{\Gamma _2} - \Omega_o R{M_m}}
\end{array}} \right\|_2^2\nonumber\\
 =& \left\| \left( {\begin{array}{*{20}{c}}
{{\Gamma _1}}\\
{{\Gamma _2}}
\end{array}} \right)  - \left( {\begin{array}{*{20}{c}}
{{\sigma _1}\sqrt {{\varepsilon _1} + {\varepsilon _3}} {N_0}{H_1}{N_m}}\\
{\Omega_o {M_m}}
\end{array}} \right)R \right\|_2^2
 +\left\| {\Gamma _1^ \bot } \right\|_2^2 + \left\| {\Gamma _2^ \bot } \right\|_2^2.\label{eq13}
\end{align}
Furthermore, we perform an inner-outer factorization, such that
\[
\left( {\begin{array}{*{20}{c}}
{{\sigma _1}\sqrt {{\varepsilon _1} + {\varepsilon _3}} {N_0}{H_1}{N_m}}\\
{\Omega_o {M_m}}
\end{array}} \right) = {\Delta _i}{\Delta _0},
\]
where ${\Delta _i}$ and ${\Delta _0}$ are the inner and the outer.
And, introduce
\[\Psi_1(s)\triangleq \left(\begin{array}{*{20}{c}}
\Delta_i^T(-s)\\
I-\Delta_i(s)\Delta_i^T(-s) \end{array} \right),\]
then, we have $\Psi_1^H(j\omega)\Psi_1(j\omega)=I$.\par
Consequently, from $(\ref{eqq8},\ref{a10},\ref{a11},\ref{eq13})$ , we have
\begin{align}
J_{12}^*=&\mathop {\inf }\limits_{K \in U} ({J_1} + {J_2})\nonumber\\
=& \left\| {\Gamma _1^ \bot } \right\|_2^2 + \left\| {\Gamma _2^ \bot } \right\|_2^2
 + \mathop {\inf }\limits_{K \in U} \left\|\Psi_1\left[{\left( {\begin{array}{*{20}{c}}
{{\Gamma _1}}\nonumber\\
{{\Gamma _2}}
\end{array}} \right) - {\Delta _i}{\Delta _0}R}\right] \right\|_2^2\nonumber\\
=& \left\| {\Gamma _1^ \bot } \right\|_2^2 + \left\| {\Gamma _2^ \bot } \right\|_2^2
+ \mathop {\inf }\limits_{K \in U}\left\| \Delta _i^{H}\left( {\begin{array}{*{20}{c}}
{{\Gamma _1}}\nonumber\\
{{\Gamma _2}}
\end{array}} \right)
- {\Delta _0}R \right\|_2^2
 + \left\|\left(I-\Delta_i\Delta_i^H\right)\left( {\begin{array}{*{20}{c}}
{{\Gamma _1}}\nonumber\\
{{\Gamma _2}}
\end{array}} \right)\right\|_2^2\nonumber
\end{align}
Similarly to $(\ref{eq9})$, $\Delta_i^H[\Gamma_1^H~~\Gamma_2^H]$ can be decomposed as
\begin{align*}
\Delta _i^{H}(s)\left( {\begin{array}{*{20}{c}}
{{\Gamma _1(s)}}\nonumber\\
{{\Gamma _2(s)}}
\end{array}} \right)=\Gamma_3^\bot(s)+\Gamma_3(s)
\end{align*}
we can design
\begin{align}\label{rc}
R=\Delta_0^{-1}\Gamma_3
\end{align}
obviously, $R\in{\mathbb{R}H_\infty}$, and $J_{12}^*$ can be written as
\begin{align}
J_{12}^*=&\left\| {\Gamma _1^ \bot } \right\|_2^2 + \left\| {\Gamma _2^ \bot } \right\|_2^2
+ \left\| {\Gamma _3^ \bot } \right\|_2^2
 + \left\|\left(1-\Delta_i\Delta_i^H\right)\left( {\begin{array}{*{20}{c}}
{{\Gamma _1}}\nonumber\\
{{\Gamma _2}}
\end{array}} \right)\right\|_2^2\nonumber\\
=&\sum\limits_{i = 1}^{N_{z}+N_{f_1}+N_{f_2}}\sum\limits_{d = 1}^{m_{i}} {\frac{r_{zid}}{(d - 1)!}}\sum\limits_{j = 1}^{N_{z}+N_{f_1}+N_{f_2}}
\sum\limits_{k = 1}^{m_{j}}{\frac{{\rm{d}}^{d-1}}{{\rm{d}}{s}^{d-1}}\frac{(-1)^{k-1}{\bar{r}_{zjk}}}{(s + {\bar{z}_{j}})^k}}\Big|_{s=z_i}\\
& + \sum\limits_{i = 1}^{N_{p}}\sum\limits_{d = 1}^{n_{i}} {\frac{r_{pid}}{(d - 1)!}}\sum\limits_{j = 1}^{N_{p}}\sum\limits_{k = 1}^{n_{j}}{\frac{{\rm{d}}^{d-1}}{{\rm{d}}{s}^{d-1}}\frac{(-1)^{k-1}{\bar{r}_{pjk}}}{(s + {\bar{p}_{j}})^k}}\Big|_{s=p_i}\nonumber\\
& + \sum\limits_{i = 1}^{N_{s}}\sum\limits_{d = 1}^{o_{i}} {\frac{r_{sid}}{(d - 1)!}}\sum\limits_{j = 1}^{N_{s}}\sum\limits_{k = 1}^{o_{j}}{\frac{{\rm{d}}^{d-1}}{{\rm{d}}{s}^{d-1}}\frac{(-1)^{k-1}{\bar{r}_{sjk}}}{(s + {\bar{s}_{j}})^k}}\Big|_{s=s_i} + \left\|\left(I-\Delta_i\Delta_i^H\right)\left( {\begin{array}{*{20}{c}}
{{\Gamma _1}}\\
{{\Gamma _2}}
\end{array}} \right)\right\|_2^2.\label{eqq17}
\end{align}
where $s_i\in{\mathbb{C}_+}, (i=1\cdots N_s)$ are the nonminimum phase zeros of
$\Delta_i^H[\Gamma_1^H~\Gamma_2^H]^H$, $o_i$ are multiples of the nonminimum phase zeros $s_i$, and
\begin{align*}
r_{s\,id}&= \frac{1}{(o_i-d)!}\frac{{\rm{d}}^{o_i-d}}{{\rm{d}}s^{o_i-d}}\Big((s-s_{i})^{o_i}
\Delta _i^{H}\left( {\begin{array}{*{20}{c}}
{{\Gamma _1}}\nonumber\\
{{\Gamma _2}}
\end{array}} \right)\Big)\Big|_{s=s_{i}},
\end{align*}
Secondly, ${J_3^*}$ can be calculated as follow:
\begin{align*}
{J_3^*}=&\mathop {\inf }\limits_{Q \in RH_{\infty}} {J_3}\\
=& \mathop {\inf }\limits_{Q \in RH_{\infty}} \left\|\left(\begin{array}{*{20}{c}}I\\0\\0\end{array}\right)+\left( {\begin{array}{*{20}{c}}
{ - {F_1}{\hat{N}}}\\
{\sqrt {\frac{{{\varepsilon _2}}}{{{\varepsilon _1}}}} {M_m}}\\[5pt]
{\sqrt {\frac{{{\varepsilon _3}}}{{{\varepsilon _1}}}} {F_{1_0}}{N_0}}
\end{array}} \right)Q \right\|_2^2{\varepsilon _1}{\sigma_r ^2}
 - {\varepsilon _2}{\Gamma _u} - {\varepsilon _3}{\Gamma _y}
\end{align*}
\begin{align*}
=& \mathop {\inf }\limits_{Q \in RH_{\infty}} \Bigg\| \left( {\begin{array}{*{20}{c}}
{L_{{f_1}}^{ - 1}L_g^{ - 1} - I}\\
0\\
0
\end{array}} \right) + \left( {\begin{array}{*{20}{c}}
I\\
0\\
0
\end{array}} \right)
 + \left( {\begin{array}{*{20}{c}}
{- {F_{{1_0}}}{N_0}}\\
{\sqrt {\frac{{{\varepsilon _2}}}{{{\varepsilon _1}}}} {M_m}}\\
{\sqrt {\frac{{{\varepsilon _3}}}{{{\varepsilon _1}}}} {F_{{1_0}}}{N_0}}
\end{array}} \right)Q \Bigg\|_2^2{\varepsilon _1}{\sigma_r ^2}
- {\varepsilon _2}{\Gamma _u} - {\varepsilon _3}{\Gamma _y}\\
=& \mathop {\inf }\limits_{Q \in RH_{\infty}} \left\|  \left( {\begin{array}{*{20}{c}}
\sqrt{\varepsilon _1}\\
0\\
0
\end{array}} \right)
+ \left({\begin{array}{*{20}{c}}
{ -\sqrt {\varepsilon _1}{F_{{1_0}}}{N_0}}\\
{\sqrt {\varepsilon _2} {M_m}}\\
{\sqrt {\varepsilon _3} {F_{{1_0}}}{N_0}}
\end{array}} \right)Q \right\|_2^2{\sigma_r ^2}
 + \left\| {\begin{array}{*{20}{c}}
{L_{{f_1}}^{ - 1}L_g^{ - 1} - 1}\\
0\\
0
\end{array}} \right\|{\varepsilon _1}{\sigma_r ^2}
 - {\varepsilon _2}{\Gamma _u} - {\varepsilon _3}{\Gamma _y}\\
=& 2{\varepsilon _1}{\sigma_r ^2}\left( \sum\limits_{i = 1}^{N_z+N_{{f_1}}} {{\mathop{\rm Re}\nolimits} \left\{ {{z_{i}}} \right\}} \right) - {\varepsilon _2}{\Gamma _u} - {\varepsilon _3}{\Gamma _y}
 + \mathop {\inf }\limits_{Q \in RH_{\infty}} {\hat{J}_3},
\end{align*}
where
\begin{align*}
{\hat{J}_3}=\left\|  \left( {\begin{array}{*{20}{c}}
\sqrt{\varepsilon _1}\\
0\\
0
\end{array}} \right)
+ \left({\begin{array}{*{20}{c}}
{ -\sqrt {\varepsilon _1}{F_{{1_0}}}{N_0}}\\
{\sqrt {\varepsilon _2} {M_m}}\\
{\sqrt {\varepsilon _3} {F_{{1_0}}}{N_0}}
\end{array}} \right)Q \right\|_2^2{\sigma_r ^2}.
\end{align*}
We introduce an inner-outer factorization such that
\begin{align*}
\left({\begin{array}{*{20}{c}}
{ -\sqrt {\varepsilon _1}{F_{{1_0}}}{N_0}}\\
{\sqrt {\varepsilon _2} {M_m}}\\
{\sqrt {\varepsilon _3} {F_{{1_0}}}{N_0}}
\end{array}} \right)=\Lambda_i\Lambda_0.
\end{align*}
And, introduce
\[\Psi_2(s)\triangleq \left(\begin{array}{*{20}{c}}
\Lambda_i^T(-s)\\
I-\Lambda_i(s)\Lambda_i^T(-s) \end{array} \right),\]
then, we have $\Psi_2^H(j\omega)\Psi_2(j\omega)=I$. It follows that
\begin{align*}
\hat{J}_3^*=&\mathop {\inf }\limits_{K \in U} {\hat{J}_3}\\
=&\mathop {\inf }\limits_{K \in U} \left\|\Lambda_i^H
\left( {\begin{array}{*{20}{c}}
\sqrt{\varepsilon_1}\\
0\\
0
\end{array}} \right)+\Lambda_0Q\right\|_2^2\sigma_r^2
+ \left\|(I-\Lambda_i\Lambda_i^H)\left( {\begin{array}{*{20}{c}}
\sqrt{\varepsilon_1}\\
0\\
0
\end{array}} \right)\right\|_2^2\sigma_r^2\\
=&\mathop {\inf }\limits_{K \in U} \left\|-\sqrt{\varepsilon_1}\Lambda_0^{-H}N_0^HF_{10}^H +\Lambda_0Q\right\|_2^2\sigma_r^2
+\left\|(I-\Lambda_i\Lambda_i^H)\left( {\begin{array}{*{20}{c}}
\sqrt{\varepsilon_1}\\
0\\
0
\end{array}} \right)\right\|_2^2\sigma_r^2.
\end{align*}
We can design
\begin{align}\label{qc}
Q=\sqrt{\varepsilon_1}(\Lambda_0^{H}\Lambda_0)^{-1}N_0^HF_{1_0}^H
\end{align}
obviously, $Q\in{\mathbb{R}H_\infty}$, and we have
\begin{align*}
\hat{J}_3^*=&\left\|\begin{array}{*{20}{c}}
\sqrt{\varepsilon_1}(I-\varepsilon_1F_{1_0}N_0\Lambda_0^{-1}\Lambda_0^{-H}N_0^HF_{1_0}^H)\\
\varepsilon_1\sqrt{\varepsilon_2}M_m\Lambda_0^{-1}\Lambda_0^{-H}N_0^HF_{1_0}^H\\
\varepsilon_1\sqrt{\varepsilon_3}F_{1_0}N_0\Lambda_0^{-1}\Lambda_0^{-H}N_0^HF_{1_0}^H
\end{array}\right\|_2^2\sigma_r^2.
\end{align*}
Therefore,
\begin{align}\label{eqq18}
J_3^*=&2{\varepsilon _1}{\sigma_r ^2}\left( \sum\limits_{i = 1}^{N_z+N_{{f_1}}} {{\mathop{\rm Re}\nolimits} \left\{ {{z_{i}}} \right\}} \right) - {\varepsilon _2}{\Gamma _u} - {\varepsilon _3}{\Gamma _y}+\left\|\begin{array}{*{20}{c}}
\sqrt{\varepsilon_1}(I-\varepsilon_1F_{1_0}N_0\Lambda_0^{-1}\Lambda_0^{-H}N_0^HF_{1_0}^H)\\
\sqrt{\varepsilon_2}M_m\Lambda_0^{-1}\Lambda_0^{-H}N_0^HF_{1_0}^H\\
\sqrt{\varepsilon_3}F_{1_0}N_0\Lambda_0^{-1}\Lambda_0^{-H}N_0^HF_{1_0}^H
\end{array}\right\|_2^2\sigma_r^2.
\end{align}
From (\ref{eqq8}), (\ref{eqq17}) and (\ref{eqq18}), we can obtain $J^*$.
\end{proof}


\begin{remark}
By the given methods in this paper, the quantitative relation between the noise variances and the tracking performance is given by the implicit results in Theorem 1. In up-link and down-link channels, the communication noise is considered simultaneously, the inner-outer factorization is presented in (\ref{tj13}) in order to design the unified controller parameter $R$, which led to only implicit relations about the noise variances and the tracking performance can be obtained.
\end{remark}
\begin{remark}
Theorem 1 assumes that the close-loop system is stable, which implies that the channel input power cannot be too small. Thus, Theorem 1 is deduced on the premise that the channel input power is large enough to ensure the stability of the close-loop system. The following Theorem 2 presents the minimum channel's input power constraints.
\end{remark}
It is known that the channel's input power constraints can't be too small, otherwise the tracking system will be unstable. The estimation of the minimum channel's input power constraints is very important and necessary. We easily obtain the following theorem by proof of the theorem 1. The minimum channel's input power constraints are given in the following theorem.
\begin{theorem}
When ensuring the stability of the system and acquiring system performance limitation, the channel's input power constraints should be satisfied
\begin{align*}
\Gamma_y\geq&\|F_{1_0}N_0(\Lambda_0^H\Lambda_0)^{-1}F_{1_0}^H\|_2^2\sigma_r^2
+\|N_0(XL^{-1}-\Delta_0^{-1}\Gamma_3F_{1_0}F_{2_0}N_0)H\|_2^2\sigma_1^2\\
&+\|F_{1_0}N_0(YB^{-1}-\Delta_0^{-1}\Gamma_3M_m)H_2\|_2^2\sigma_2^2,\\
\Gamma_u\geq&\|M_m(\Lambda_0^H\Lambda_0)^{-1}F_{1_0}^H\|_2^2\sigma_r^2
+\|F_{2_0}N_0(YB^{-1}-\Delta_0^{-1}\Gamma_3M_m)H_1\|_2^2\sigma_1^2\\
&+\|M_m(YB^{-1}-\Delta_0^{-1}\Gamma_3M_m)H_2\|_2^2\sigma_2^2.
\end{align*}
\end{theorem}
\begin{proof}
From (\ref{eq7}) and (\ref{rc}), we can obtain
\begin{align*}
\nonumber{E\bf \|y\|_2^2}=& E\|P{\left( {I - {K_2}{F_2}P{F_1}} \right)^{ - 1}}
\left( {{F_1}{K_1}{\bf r} + {H_1}{{\bf n}_1} + {F_1}{K_2}{H_2}{{\bf n}_2}} \right)\|_2^2,\\
=&\nonumber  \left\| {\hat N\left( {X - RN} \right){H_1}{\sigma _1}} \right\|_2^2
 + \Big\| {F_1}\hat N( Y - RM ){H_2}{\sigma _2} \Big\|_2^2  + \Big\| {{F_1}\hat NQ\sigma_r } \Big\|_2^2\\
=&\|F_{1_0}N_0(\Lambda_0^H\Lambda_0)^{-1}F_{1_0}^H\|_2^2\sigma_r^2
+\|N_0(XL^{-1}-\Delta_0^{-1}\Gamma_3F_{1_0}F_{2_0}N_0)H\|_2^2\sigma_1^2\\
&+\|F_{1_0}N_0(YB^{-1}-\Delta_0^{-1}\Gamma_3M_m)H_2\|_2^2\sigma_2^2,
\end{align*}
Similarly, from (\ref{eq7}), (\ref{rc}) and (\ref{qc}), we can obtain
\begin{align*}
\nonumber{E\bf \|u\|_2^2}=& E\|{\left( {I - {K_2}{F_2}P{F_1}} \right)^{ - 1}}
\left( {{K_1}{\bf r} + {K_2}{F_2}P{H_1}{{\bf n}_1} + {K_2}{H_2}{{\bf n}_2}} \right)\|_2^2,\\
=& \left\| {{F_2}\hat N\left( {Y - RM} \right){H_1}{\sigma _1}} \right\|_2^2
 + \Big\| {M\left( {Y - RM} \right){H_2}{\sigma _2}} \Big\|_2^2 + \Big\| {MQ\sigma_r } \Big\|_2^2\\
=&\|M_m(\Lambda_0^H\Lambda_0)^{-1}F_{1_0}^H\|_2^2\sigma_r^2
 +\|F_{2_0}N_0(YB^{-1}-\Delta_0^{-1}\Gamma_3M_m)H_1\|_2^2\sigma_1^2\\
 &+\|M_m(YB^{-1}-\Delta_0^{-1}\Gamma_3M_m)H_2\|_2^2\sigma_2^2.
\end{align*}
\end{proof}

\section{Simulation Studies}\label{se4}
In this section, some examples are given to show the effectiveness of the obtained theoretical result, which is also used to analyze the performance for a real-time random-noise tracking radar system\cite{Cover2012,Zhang2004}. To better present the impact of different channel factors on tracking performance, the up-link channel and down-link channel are considered in Example 1 and Example 2, respectively. Because weight factors $\varepsilon_1, \varepsilon_2$ and $\varepsilon_3$ are used to measure influence level of tracking error, down-link channel, or up-link channel, therefore, conditions of $\varepsilon_2 = 0$ and $\varepsilon_3 = 0$ are considered in Examples 1 and 2, respectively. The minimum channel input power constraints are analyzed in Example 3.
\par
{\bf Example 1:} Consider a continuous plant with its transfer function given by
$$P(s)=\frac{s-k}{(s+1)(s-p)}.$$
Clearly, $P(s)$ is non-minimum phase and unstable for $p>0$ and $k>0$. The following will consider the case of a up-link channel. The LTI filters used to model the finite bandwidth $F_1(s)=1, F_2(s)=f_2/(s+f_2)$ and colored noise $H_1(s)=0, H_2(s)=h_2/(s-h_2)$ of the communication link are both chosen to be low-pass Butterworth filters of order 1. The system output channels power constraint $\Gamma_y=2.5$ and $\varepsilon_2=0$.
\par
Two observations can be obtained from Fig.\ref{Fig.8}, where the optimal performance is plotted with respect to bandwidth of both $F(s)$ and $H(s)$. First, system tracking performance becomes better as the available bandwidth of the communication channel decreases. Secondly, if the noise is colored by a low-pass filter, the decrease of its cutoff frequency would lead to the better tracking performance. Fig.\ref{Fig.4} and Fig.\ref{Fig.5} show the optimal performances plotted with respect to unstable pole $p$ and NMP zero $k$ for different values of $\varepsilon_1$. It can be observed from Fig.\ref{Fig.4} and Fig.\ref{Fig.5} that unstable pole and NMP zero worsen tracking performance like the way demonstrated in
Theorem 1. Besides, Fig.\ref{fig4a} and Fig.\ref{fig5a} show that, when nonminimum phase zero and unstable zero are located closely, the performance will be badly degraded. Additionally, Fig.\ref{fig4b} and Fig.\ref{fig5b} show that, when pole-zero cancellation does not occur, the impact on performance by the unstable pole or NMP zero will become more intense.
Fig.\ref{Fig.6} and Fig.\ref{Fig.7} shows that the reference signal and the noise signal will deteriorate tracking performance.

\begin{figure}[H]
\centering
\begin{minipage}[c]{0.7\textwidth}
\centering
  \includegraphics[height=5.3cm,width=8cm]{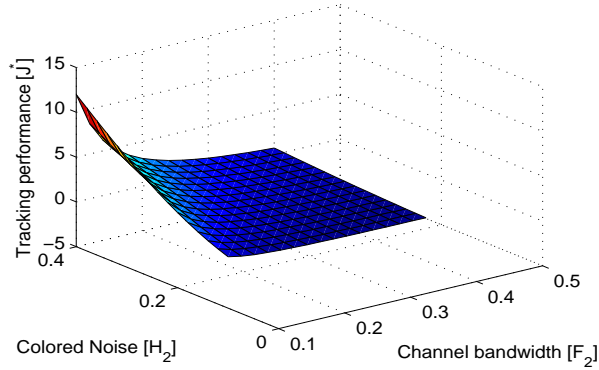}
  \caption{ $J^*$ with respect to channel bandwidth $F_2$ and colored noise $H_2$.  \protect\\ $(k=3, p=2, \sigma_r=0.2, \sigma_2=0.1, \varepsilon_1=0.5,\varepsilon_3=0.5)$}\label{Fig.8}
\end{minipage}
\end{figure}

\begin{figure}[H]
\centering
\subfloat[]{
\label{fig4a}
\begin{minipage}[t]{0.45\textwidth}
\centering
\includegraphics[height=5.3cm,width=8cm]{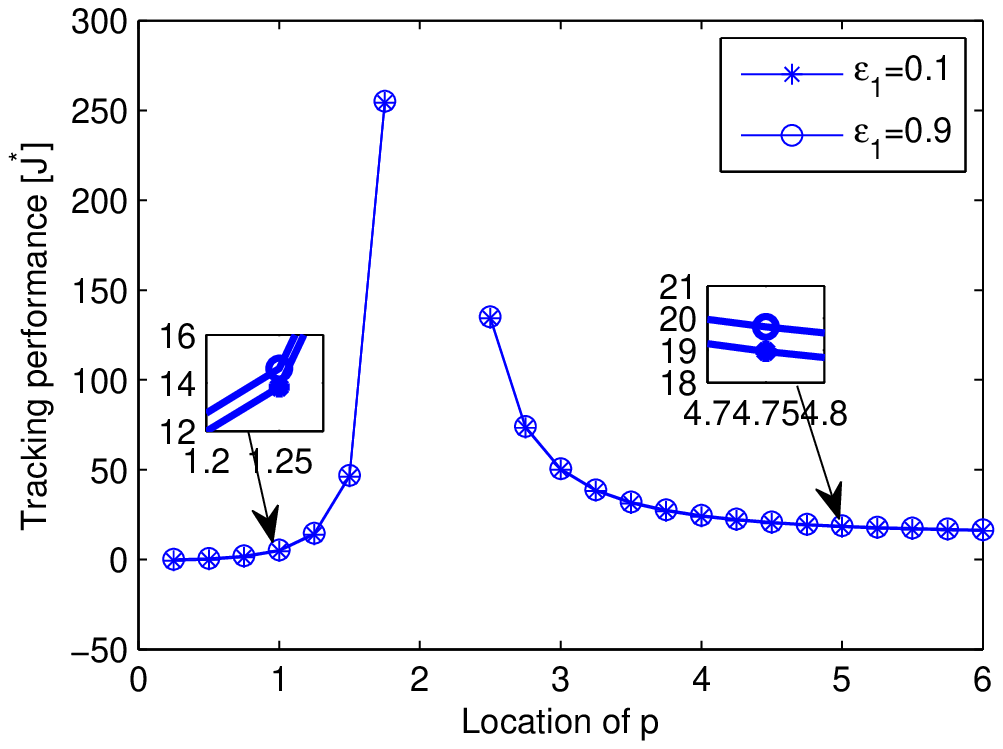}
\end{minipage}
}
\subfloat[]{
\label{fig4b}
\begin{minipage}[t]{0.45\textwidth}
\centering
\includegraphics[height=5.3cm,width=8cm]{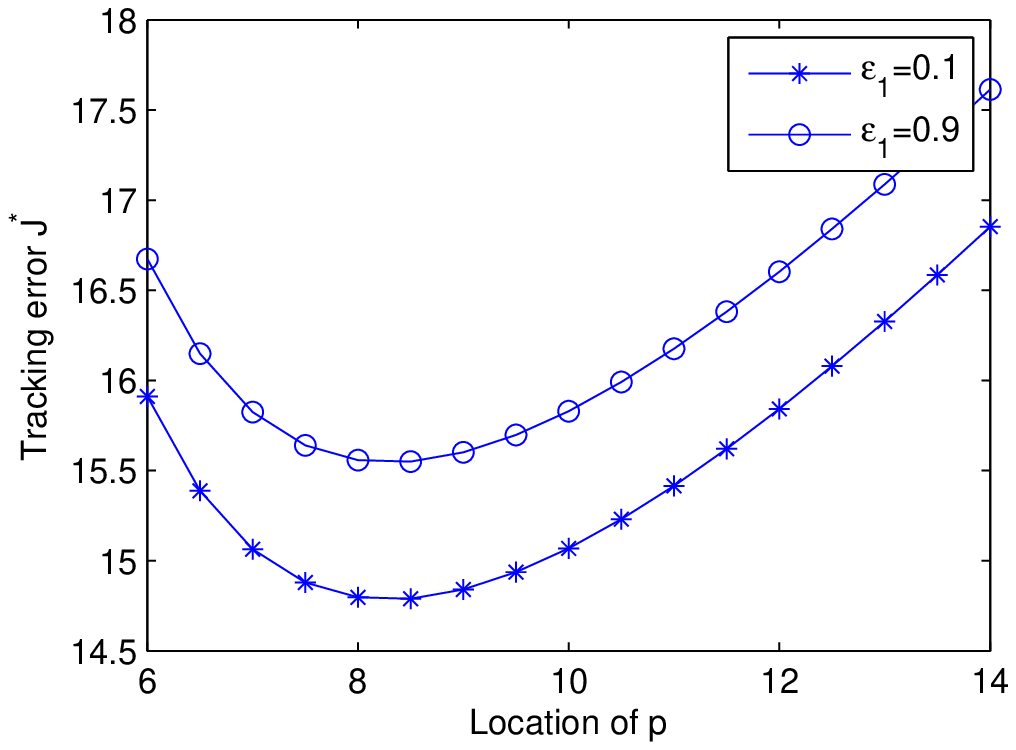}
\end{minipage}
}
\caption{$J^*$ with respect to $p$ for different $\varepsilon_1$ and $\varepsilon_3$.  \protect\\ $(k=2, \sigma_r=0.1, \sigma_2=0.3, f_2=0.1, h_2=0.2, \varepsilon_1+\varepsilon_3=1)$}\label{Fig.4}
\end{figure}

\begin{figure}[H]
\centering
\subfloat[]{
\label{fig5a}
\begin{minipage}[t]{0.45\textwidth}
\centering
\includegraphics[height=5.3cm,width=8cm]{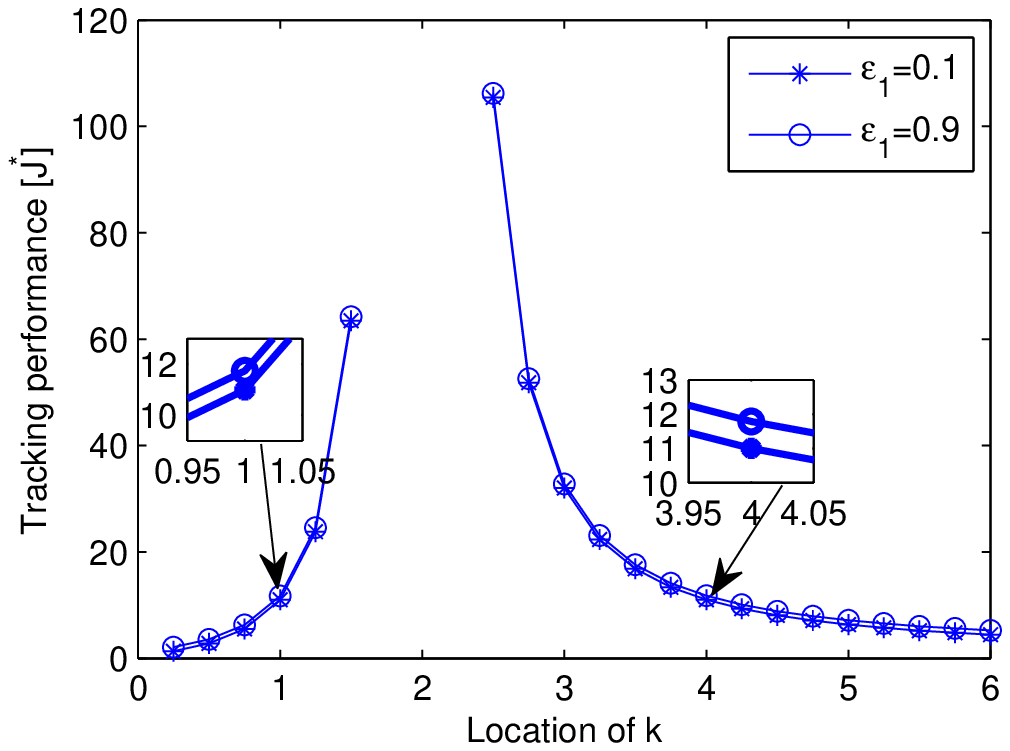}
\end{minipage}
}
\subfloat[]{
\label{fig5b}
\begin{minipage}[t]{0.5\textwidth}
\centering
\includegraphics[height=5.3cm,width=8cm]{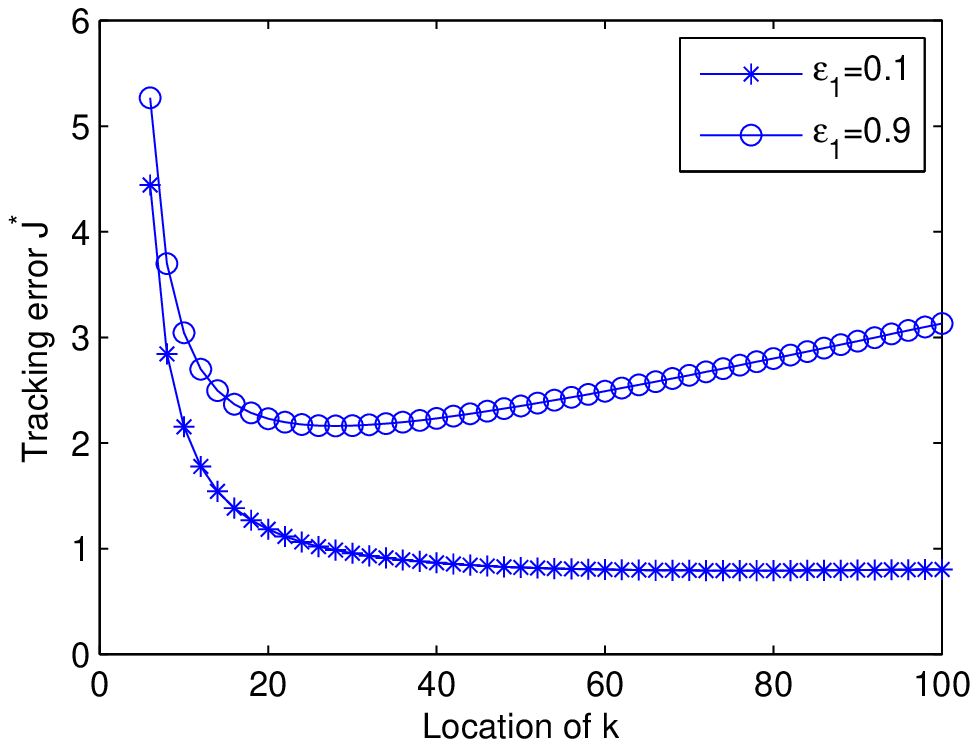}
\end{minipage}
}
\caption{ $J^*$ with respect to $k$ for different $\varepsilon_1$ and $\varepsilon_3$. \protect\\ $(p=2, \sigma_r=0.1, \sigma_2=0.3, f_2=0.1, h_2=0.2, \varepsilon_1+\varepsilon_3=1)$ }\label{Fig.5}
\end{figure}

\begin{figure}[H]
\centering
\begin{minipage}[c]{0.7\textwidth}
\centering
  \includegraphics[height=5.5cm,width=8cm]{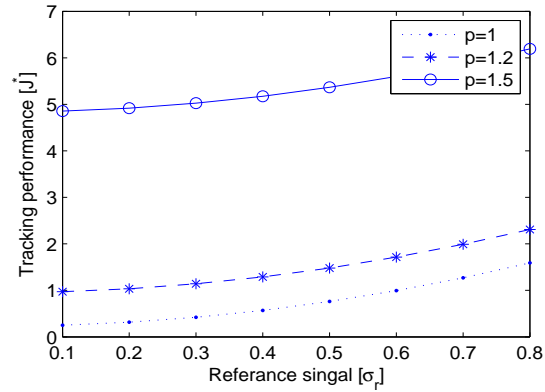}
  \caption{ $J^*$ with respect to $\sigma_r$ for different unstable pole $p$. \protect\\ $(k=2, \sigma_2=0.1, f_2=0.1,$ $h_2=0.2, \varepsilon_1=0.5,\varepsilon_3=0.5)$ }\label{Fig.6}
\end{minipage}
\end{figure}

\begin{figure}[H]
\centering
\begin{minipage}[c]{0.7\textwidth}
\centering
  \includegraphics[height=5.5cm,width=8cm]{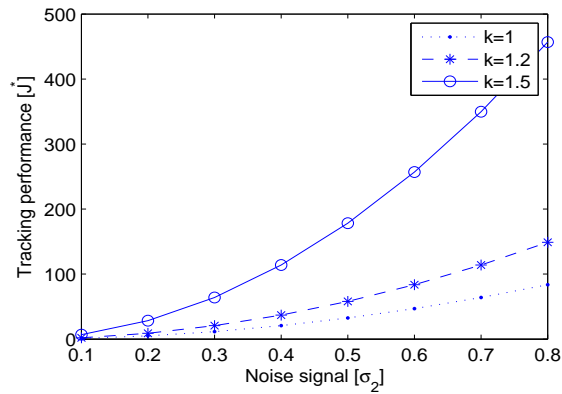}
  \caption{ $J^*$ with respect to $\sigma_2$ for different NMP zero $k$. \protect\\ $(p=2, \sigma_r=0.2, f_2=0.1,$ $h_2=0.2, \varepsilon_1=0.5,\varepsilon_3=0.5)$ }\label{Fig.7}
\end{minipage}
\end{figure}
\par
{\bf Example 2:} The following will consider the case of an down-link channel. The plant is non-minimum phase and unstable with two NMP zeros and two unstable poses. The continuous plant with its transfer function is given by
$$P(s)=\frac{(s-k_1)(s-k_2)}{(s+1)(s-p_1)(s-p_2)}.$$
$P(s)$ is non-minimum phase and unstable for $p_1>0, p_2>0, k_1>0, k_2>0$. The LTI filters used to model the finite bandwidth $F_1(s)=f_1/(s+f_1), F_2(s)=1$ and colored noise $H_1(s)=h_1/(s-h_1), H_2(s)=0$ of the communication link are both chosen to be low-pass Butterworth filters of order 1. The system output channels power constraint $\Gamma_u=2.5$ and $\varepsilon_3=0$. \par
\begin{figure}[H]
\centering
\begin{minipage}[c]{0.7\textwidth}
\centering
  \includegraphics[width=8.4cm,height=6cm]{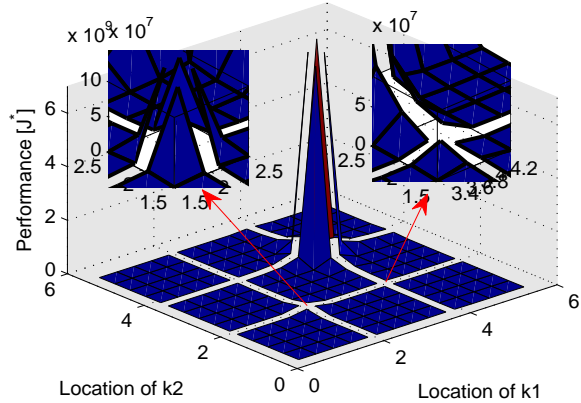}
  \caption{ $J^*$ with respect to NMP zeros $k_1$ and $k_2$.
  $(0.2<k_1,k_2<5.5, p_1=2, p_2=3.8, \sigma_r=0.5, \sigma_1=0.5,
  $ $ \varepsilon_1=0.5, \varepsilon_2=0.5)$}\label{Fig.9}
\end{minipage}
\end{figure}
\begin{figure}[H]
\centering
\begin{minipage}[c]{0.7\textwidth}
\centering
  \includegraphics[width=8.4cm,height=6cm]{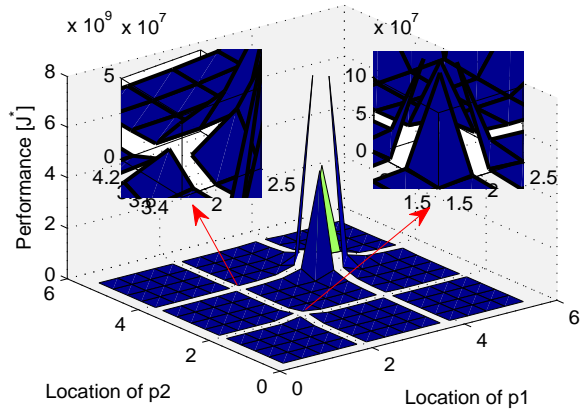}
  \caption{ $J^*$ with respect to unstable poles $p_1$ and $p_2$.
  $(0.2<p_1,p_2<5.5, k_1=2, k_2=3.8, \sigma_r=0.5, \sigma_1=0.5,
   $ $ \varepsilon_1=0.5, \varepsilon_2=0.5)$}\label{Fig.10}
\end{minipage}
\end{figure}

Fig.\ref{Fig.9} and Fig.\ref{Fig.10} can be obtained. Fig.\ref{Fig.9}/Fig.\ref{Fig.10} shows the relationship between the location of the NMP zeros/unstable poles and the optimal tracking performance. Fig.\ref{Fig.10} presents the relationship between the location of the unstable pole and the optimal tracking performance. It can be found that in Fig.\ref{Fig.9} and Fig.\ref{Fig.10}, optimal tracking performance tends to be infinite when pole-zero cancellation takes place. Another phenomenon is revealed in Figs.\ref{Fig.9} and Figs.\ref{Fig.10}, when two poles/two zeros cancellation takes place in the plant, the optimal tracking performance is more severely deteriorated than when one pole-zero cancellation takes place in the plant.

{\bf Example 3:} The following will consider the case of a up-link channel.
The LTI filters used to model the finite bandwidth $F_1(s)=1, F_2(s)=1/(s+1)$ and colored noise $H_1(s)=0, H_2(s)=1/(s-1)$ of the communication link are both chosen to be low-pass Butterworth filters of order 1. And, $p_1=12,p_2=0,k_1=15,k_2=0,\varepsilon_1=0.8,\varepsilon_2=0.2, \varepsilon_3=0, \Gamma_y=1, 0.1<\sigma_r,\sigma_2<0.5$. From Theorem 2, we have $\Gamma_y\geq\max\{{2.88\sigma_2^2+0.64\sigma_r^2}\}$, thus, the performance limitation with channel input constraint can be achieved under $\max\{{2.88\sigma_2^2+0.64\sigma_r^2}\}\leq{1}$. In this case, Fig.\ref{Fig.11} can be obtained.
\par
The relationship among the reference signal, channel noise and the optimal tracking performance is shown in Fig.\ref{Fig.11}. However, owing to the channel input power constraint, the performance limitation can be obtained only in the left part of the Fig.\ref{Fig.11}.

\begin{figure}[H]
\centering
\begin{minipage}[c]{0.8\textwidth}
\centering
  \includegraphics[width=8.4cm,height=6cm]{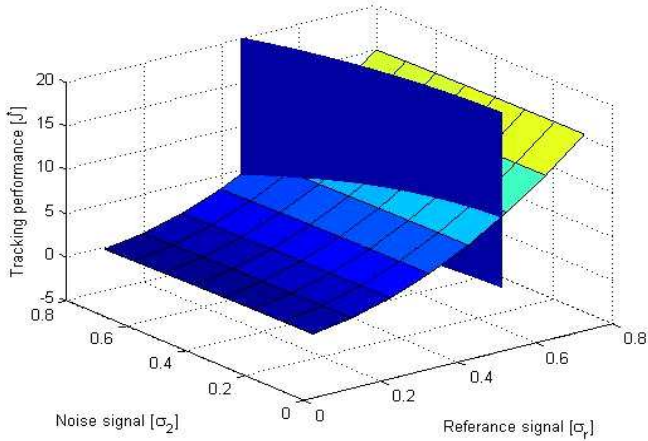}
  \caption{ $J^*$ with respect to reference $r$ and the channel noise $n_2$.}\label{Fig.11}
\end{minipage}
\end{figure}

\section{Conclusions}\label{se5}

In this paper, we have investigated the optimal tracking performance of control systems under up-link and down-link channels with channel input constraint on the power. The limited bandwidth and additive colored Gaussian noise is considered in communication channels. And a two-parameter controller is adopted. We have derived explicit expressions for constrained optimal tracking performance using $\mathcal{H}_2$ optimization techniques. The results show that, the optimal tracking performance depends on characteristics of the system and the up-link and down-link channels. Furthermore, due to the existence of the network, the best achievable tracking performance will also be adversely affected by the limited bandwidth, the input power constraints and additive colored Gaussian noises of the communication channel. Additionally, the channel minimal input power constraints are given under the condition ensuring the stability of the system and acquiring system performance limitation. Besides, some simulation results are given to illustrate the obtained results.
\par
The current work can be extended to deal with the performance issues over more complex network environment. Although much more complicated, it is interesting to derive similar results for multivariable plants with wireless networks in up/down link channels.

\end{document}